\documentclass[preprint,authoryear,12pt]{elsarticle}

\usepackage[ruled]{algorithm2e}

\SetAlFnt{\small}
\SetAlCapFnt{\small}
\SetAlCapNameFnt{\small}
\SetAlCapHSkip{0pt}
\IncMargin{-\parindent}

\usepackage{amssymb}
\usepackage{amsfonts}
\usepackage{amsmath}
\usepackage{natbib}

\usepackage{amsthm}
\usepackage{amsfonts}

\newtheorem{theorem}{Theorem}
\newtheorem{lemma}[theorem]{Lemma}

\newtheorem{definition}[theorem]{Definition}

\def\FPRAS{\mathrm{FPRAS}}
\def\FPAUS{\mathrm{FPAUS}}

\def\MPSCJ{\mathrm{MPSCJ}}

\def\sharpMPSCJ{\mathrm{\#MPSCJ}}
\def\SPSCJ{\mathrm{SPSCJ}}
\def\sharpSPSCJ{\mathrm{\#SPSCJ}}
\def\sharpFSPSCJ{\mathrm{\#Fitch-SPSCJ}}

\def\DCJ{\mathrm{DCJ}}
\def\SCJ{\mathrm{SCJ}}

\def\P{\mathrm{P}}
\def\NP{\mathrm{NP}}
\def\NPcomp{\mathrm{NP-complete}}
\def\NPhard{\mathrm{NP-hard}}
\def\FP{\mathrm{FP}}
\def\sharpP{\mathrm{\#P}}
\def\sharpPcomp{\mathrm{\#P-complete}}
\def\RP{\mathrm{RP}}

\def\FPAUS{\mathrm{FPAUS}}
\def\SAT{\mathrm{3SAT}}
\def\sharp3SAT{\mathrm{\#3SAT}}
\def\CNF{\mathrm{3CNF}}
\def\root{\mathrm{root}}

\journal{Theoretical Computer Science}
\begin{document}

\begin{frontmatter}


\title{On sampling SCJ rearrangement scenarios}
\author{Istv\'an Mikl\'os$^{1,2}$ and
S\'andor Z. Kiss$^{2,3}$ and
Eric Tannier$^4$}
\address{$^{1}$Department of Stochastics, R{\'e}nyi Institute, 1053
  Budapest, Re{\'a}ltanoda u. 13-15, Hungary\\
$^{2}$ Data Mining and Search Research Group, Computer and Automation
Institute, Hungarian Academy of Sciences, 1111 Budapest,
L\'agym\'anyosi \'ut 11, Hungary\\
$^3$ University of Technology and Economics, Department of Algebra,
1111 Budapest, Egry J\'ozsef utca 1, Hungary \\
$^{4}$ INRIA Rh\^one-Alpes ; Universit\'e de Lyon ; Universit\'e Lyon 1 ; CNRS, UMR5558, Laboratoire de Biom\'etrie et Biologie \'Evolutive, F-69622, Villeurbanne, France.}

\begin{abstract}
The Single Cut or Join ($\SCJ$) operation on genomes,
generalizing chromosome evolution by fusions and fissions,
is the computationally simplest known model
of genome rearrangement. While
most genome rearrangement problems are already hard when comparing
three genomes, it is possible to compute in polynomial time
a most parsimonious $\SCJ$ scenario for an arbitrary
number of genomes related by a binary phylogenetic tree.

Here we consider the problems of sampling and counting the most
parsimonious $\SCJ$ scenarios. We show that both the sampling and
counting problems are easy for two genomes, and we relate $\SCJ$
scenarios to alternating permutations.
However, for an arbitrary number of genomes related by a binary phylogenetic tree,
the counting and sampling problems become hard. We prove that if a
Fully Polynomial Randomized Approximation Scheme or a Fully Polynomial
Almost Uniform Sampler exist for the most parsimonious $\SCJ$
scenario, then $\RP = \NP$.

The proof has a wider scope than genome rearrangements: the same result
holds for parsimonious evolutionary scenarios on any set of discrete characters.
\end{abstract}



\begin{keyword}
MSC codes: F.2.2: Computations on discrete structures \sep G.2.1:
  Counting problems \sep free keywords:
  Single cut and join \sep FPAUS \sep FPRAS \sep non-approximability
\end{keyword}





\end{frontmatter}

\section{Introduction}

\label{sec:introduction}

The genome rearrangement problem is one of the oldest optimization problems in
computational biology. It has been already formulated by
\cite{sn1941}. It consists in finding the minimum number of
rearrangement events that can explain the gene order differences between
two genomes.
According to how genomes and rearrangements are defined, a
number of variants have been studied \citep{Fertin2009}. In many
cases, efficient algorithms running in polynomial time exist for
finding one solution, but they do not scale up to three genomes:
finding a median, {\em i.e.}, a genome minimizing the sum of
the number of rearrangements to the three others, is almost
always $\NPhard$.

Moreover, one solution is not representative of the whole
optimal solution space. So another computational problem
is to find all minimum solutions.
But the number of minimum solutions is often so high that
their explicit enumeration is not possible in polynomial running
time. A small number of samples coming from (almost) the uniform
distribution is usually sufficient for testing evolutionary hypotheses
like the Random Breakpoint Model \citep{Alekseyev2010,Bergeron2008} or
the sizes and positions of inversions
\citep{Ajana2002,dmr2008}. Drawing conclusions from one scenario or
from a biased sample should be avoided as it might be very misleading
\citep{Bergeron2008,md2009}. 

Statistical methods, like Markov chain Monte Carlo methods, can
sample genome rearrangement scenarios \citep{dmr2008,durrett2004,larget2002,larget2005,Miklos2010},
but often there are no available results for their mixing time.
Only in the case of the Double Cut-and-Join ($\DCJ$)
rearrangement model, a Fully Polynomial
time Randomized Approximation Scheme ($\FPRAS$) and a Fully Polynomial
Almost Uniform Sampler ($\FPAUS$) are available for counting and sampling
most parsimonious rearrangement scenarios between two genomes \citep{mt2012}.
But this is hardly generalizable to more than two genomes because
for $\DCJ$ the median problem is $\NPhard$ \citep{tzs2009}.

Recently, a simpler rearrangement model has been published
by \cite{fm2011} under the name {\em Single Cut or Join}, or $\SCJ$.
It consists in a gain and loss process on gene {\em adjacencies},
and from a chromosomal point of view, allows fusions and fissions,
linearization of circular chromosomes and vice versa.
The computational simplicity of this model is highlighted 
by the existence of an easy polynomial running time algorithm
for the median problem. More generally, finding a most
parsimonious $\SCJ$ scenario on an arbitrary evolutionary tree (the small
parsimony problem) is also polynomial.

Therefore, it is reasonable to assume that
at least stochastic approximations are available for the number of most
parsimonious $\SCJ$ scenarios. We show here that it is the case for two
genomes. However, we report a negative
result for the small parsimony problem: the number of
most parsimonious $\SCJ$ scenarios cannot be approximated in
polynomial time even in a stochastic manner unless $\RP = \NP$.
This bounds the possibilities of using this model for genomic
studies.

The paper is organized as follows. The next Section formally
introduces useful vocabulary in genome rearrangement and random algorithm
complexity. In Section~\ref{sec:twogenomes} we show that counting
and sampling $\SCJ$ scenarios between two genomes is easy,
and show the relation with the so-called Andr\'e's problem on
alternating permutations.
The hardness theorems for an arbitrary number of genomes
are stated and proved in Section~\ref{sec:arbitrarytree}.
The paper ends with a discussion on the impact of these results and the
statements of some related open problems. 

\section{Genome rearrangement: finding, counting, sampling}
\label{sec:definitions}

\subsection{Genome rearrangement by $\SCJ$}

\begin{definition}
A \emph{genome} is a directed, edge-labelled graph, in which each
vertex has a total degree at most 2,
 and each label is unique. Each edge is called a \emph{gene}.
The beginning of an edge is called \emph{tail}, the end of an
edge is called \emph{head}, the joint name of heads and tails
is \emph{extremities}.
The vertices with degree $2$ are called
\emph{adjacencies}, the vertices with degree $1$ are called
\emph{telomeres}.
\end{definition}

By definition, a genome is a set of disjoint paths and cycles,
and neither the paths nor the cycles
are necessarily directed. The components of the genome are the
\emph{chromosomes}.
An example for a couple of genomes is drawn on Figure~\ref{fig:genome}.
\begin{figure}[tpb]
\centerline{\includegraphics[scale=0.5]{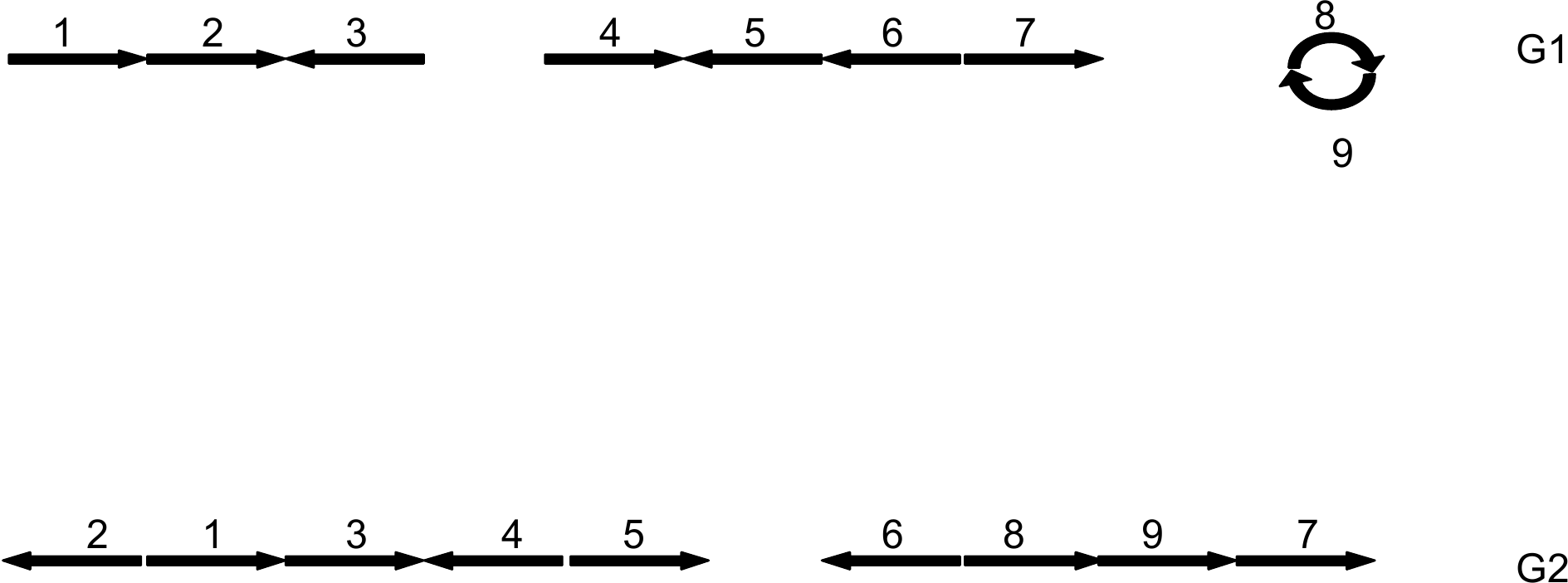}}
\caption{An example of two genomes with 9 genes.}\label{fig:genome}
\end{figure}
All adjacencies correspond to two gene extremities and
telomeres to one. For example, $(h1,t3)$ describes the vertex
of genome $G_2$ in Figure~\ref{fig:genome} in which the head of gene
$1$ and the tail of gene $3$ meet, and similarly, $(h7)$
is the telomere where gene $7$ ends.
A genome is fully described by a list of such descriptions of
adjacencies and telomeres. 

We will study several genomes simultaneously. We always assume the
genomes we compare have the same label set. It means they are required
to have exactly the same gene content.

\begin{definition}
A \emph{Single Cut or Join ($\SCJ$)}  operation transforms one genome
into another by modifying the adjacencies and telomeres
in one of the following $2$ ways:
\begin{itemize}
\item take an adjacency $(a,b)$ and replace it by two telomeres, $(a)$
  and $(b)$.
\item take two telomeres $(a)$ and $(b)$, and replace them by an
  adjacency $(a,b)$.
\end{itemize}
\end{definition}

Given two genomes $G_1$ and $G_2$, it is always
possible to transform one into the other by a sequence of $\SCJ$ operations
\citep{fm2011}. Such a sequence is called an $\SCJ$ {\em scenario}
for $G_1$ and $G_2$. Scenarios of minimum length are called
{\em most parsimonious}, and their length is
the {\em $\SCJ$ distance} and is denoted by $d_{\SCJ}(G_1,G_2)$.


The adjacency graph was introduced by \cite{bms2006} to compute
the $\DCJ$ distance between two genomes. It can be used to study
$\SCJ$ scenarios as well:

\begin{definition}
The \emph{adjacency graph} $G(V_1\cup V_2,E)$ of two genomes $G_1$ and $G_2$
is a bipartite multigraph in which $V_1$ is the set of adjacencies
and telomeres of $G_1$ and $V_2$ is the set of adjacencies and telomeres of $G_2$.
The number of edges between $u \in V_1$ and $v \in V_2$ is the number of
extremities they share.
\end{definition}

Each vertex of the adjacency graph has either degree $1$ or $2$, and
thus, the adjacency graph falls into disjoint cycles and paths. Each path
has one of the following three types:

\begin{itemize}
\item {\em odd path}, containing an odd number of edges and an even number of vertices,
\item {\em $W$-shaped path}, which is an even path with two endpoints in $V_1$
\item {\em $M$-shaped path}, which is an even path with two endpoints in $V_2$
\end{itemize}

In addition we call {\em trivial} components the cycles with
two edges and the paths with one edge. An adjacency graph example 
can be seen on Figure~\ref{fig:adjacency}.

\begin{figure}[tpb]
\centerline{\includegraphics[scale=0.5]{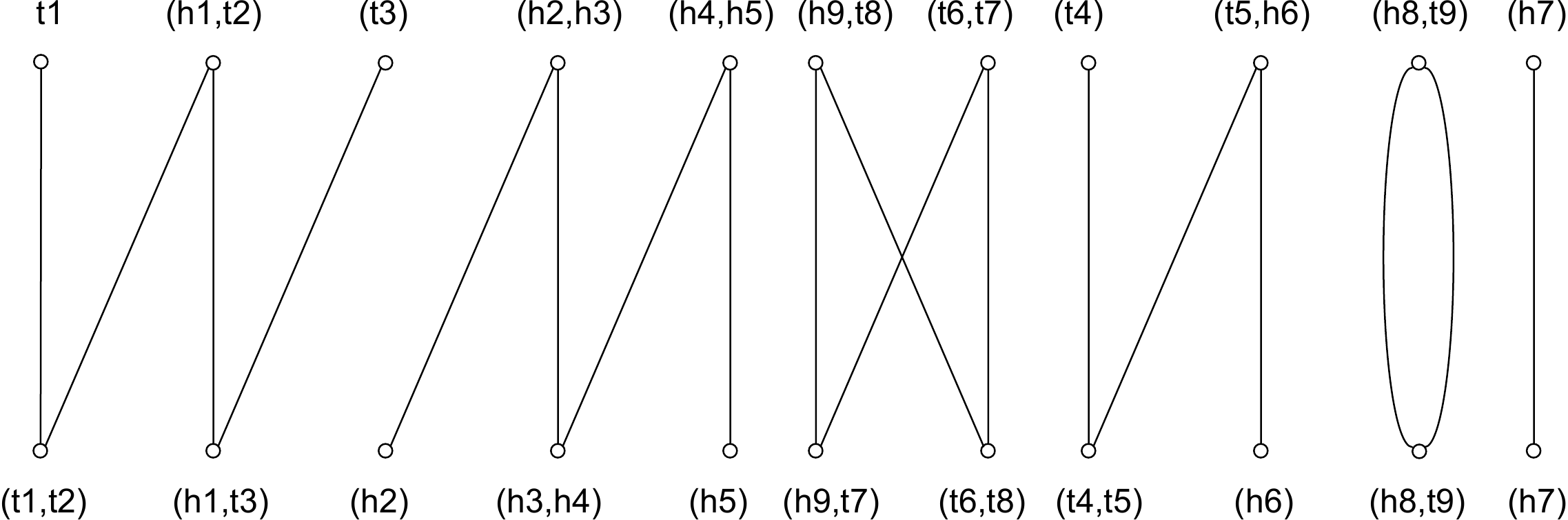}}
\caption{The adjacency graph of the two genomes on Fig.~\ref{fig:genome}}\label{fig:adjacency}
\end{figure}

\subsection{Counting and Sampling $\SCJ$ scenarios}

\begin{definition}
A decision problem is in $\NP$ if a non-deterministic Turing Machine
can solve it in polynomial time. An equivalent definition is that a
witness proving the ``yes'' answer to the question can be verified in
polynomial time. A counting problem is in $\sharpP$ if it asks for the
number of witnesses of a problem in $\NP$.
\end{definition}

\begin{definition}
A decision problem is in $\RP$ if a random algorithm exists with the
following properties: a) the running time is deterministic and grows
polynomially with the size of the input, b) if the true answer is
``no'', then the algorithm answers ``no'' with probability $1$, c) if
the true answer is ``yes'', then it answers ``yes'' with probability
at least $1/2$.
\end{definition}

\begin{definition}
The \emph{Most Parsimonious $\SCJ$ scenario problem ($\MPSCJ$)} is
to compute $d_{\SCJ}(G_1,G_2)$ for two genomes $G_1$ and $G_2$ given
as input.
The \emph{$\sharpMPSCJ$} problem asks for the number of scenarios
of length $d_{\SCJ}(G_1,G_2)$, denoted by $\sharpMPSCJ(G_1,G_2)$.
\end{definition}

For example, the $\SCJ$ distance between the two genomes of Figure~\ref{fig:genome}
is $12$ and there are $16 \times {12 \choose 3 \ 3 \  4 \  2}$
different scenarios.

$\MPSCJ$ is an optimization problem, which has
a natural corresponding decision problem asking if there is a scenario
with a given number of $\SCJ$ operations. So we may write that
$\sharpMPSCJ \in \sharpP$, which means that $\sharpMPSCJ$
asks for the number of witnesses of the
decision problem ``Is there a scenario for $G_1$ and $G_2$
of size $d_{\SCJ}(G_1,G_2)$ ?''.

\begin{definition}
Given a rooted binary tree $T(V,E)$  with $k$ leaves, and genomes $G_1,
G_2, \ldots, G_k$ assigned to the leaves, the \emph{small parsimony SCJ} problem ($\SPSCJ$)
asks for an assignment of genomes to the internal nodes of $T$ and
an $\SCJ$ scenario for each edge, which minimize the number of $\SCJ$
operations along the tree, {\em i.e.},
\begin{equation}\label{eq:spscj_score}
\sum_{(v_i,v_j) \in E} d_{\SCJ}(G_i,G_j)
\end{equation}
where $G_i$ ($G_j$) is the genome which is assigned to vertex $v_i \in
V$ ($v_j \in V$).

The \emph{small parsimony} term is borrowed from the well-known
textbook problem \citep{jp2004}, the \emph{small parsimony problem of
  discrete characters}: Given a rooted binary tree $T(V,E)$ with $k$
leaves labelled by characters from a finite alphabet, label the internal
nodes such that the number of edges labelled with different characters
at their two ends is minimized.

The solution space of the most parsimonious $\SCJ$ scenarios on a tree
consists of all possible combinations of assignments to the internal
nodes together with the possible $\SCJ$ scenarios on the edges of the
phylogenetic tree. The $\sharpSPSCJ$ problem asks the size of this
solution space.
\end{definition}

As the decision version of $\SPSCJ$ is trivially in $\NP$, $\sharpSPSCJ$
is in $\sharpP$. There are subclasses
in $\sharpP$ containing counting problems which are approximable
by polynomial deterministic or randomized algorithms.

\begin{definition}
A counting problem in $\sharpP$ is in $\FP$ if there is a polynomial running time algorithm
which gives the solution. It is $\sharpPcomp$ if any problem in $\sharpP$ 
can be reduced to it by a polynomial-time counting reduction.
\end{definition}

\begin{definition}
A counting problem in $\sharpP$ is in $\FPRAS$ (\emph{Fully Polynomial
Randomized Approximation Scheme}) if there exists a randomized
algorithm such that for any instance $x$, and $\epsilon, \delta > 0$,
it generates an approximation $\hat{f}$ for the solution $f$, satisfying
\begin{equation}
P\left(\frac{f}{1+\epsilon} \leq \hat{f} \leq f(1+\epsilon)\right) \geq 1 - \delta
\end{equation}
and the algorithm has a time complexity bounded by a polynomial of
$|x|$, $1/\epsilon$ and $-\log(\delta)$.
\end{definition}

The total variational distance $d_{TV}(p,\pi)$ between two discrete
distributions $p$ and $\pi$ over the set $X$ is defined as
\begin{equation}
d_{TV}(p,\pi) := \frac{1}{2}\sum_{x \in X} |p(x) - \pi(x)|
\end{equation}

\begin{definition}
A counting problem in $\sharpP$ is in $\FPAUS$ if there exists a randomized
algorithm (a \emph{Fully Polynomial Almost Uniform Sampler} that is
also abbreviated as $\FPAUS$) such that for any instance $x$, and $\epsilon > 0$, it
generates a random element of the solution space following a distribution $p$ satisfying
\begin{equation}
d_{TV}(p,U) \leq \epsilon
\end{equation}
where $U$ is the uniform distribution over the solution space,
and the algorithm has a time complexity bounded by a polynomial of $|x|$, and
$-\log(\epsilon)$. 
\end{definition}

\section{Most parsimonious SCJ scenarios between two genomes}
\label{sec:twogenomes}


\subsection{A dynamic programming solution}

The $\SCJ$ distance can be calculated in polynomial time, as stated in
the following theorem.
\begin{theorem}\label{theo:dscj}
 (\cite{fm2011})
Let $\Pi_1$ denote the set of adjacencies in genome $G_1$ and let
$\Pi_2$ denote the set of adjacencies in genome $G_2$. Then
\begin{equation}
d_{\SCJ}(G_1,G_2) = \left|\Pi_1 \Delta \Pi_2\right|
\end{equation}
where $\Delta$ denotes the symmetric difference of the two sets.
\end{theorem}

Theorem~\ref{theo:dscj} says that any shortest path transforming $G_1$
into $G_2$ has to cut all the adjacencies in $G_1 \setminus G_2$ and
add all the adjacencies in $G_2 \setminus G_1$, and there are no more
$\SCJ$ operations. Drawing one solution is easy: first cut all adjacencies in $G_1 \setminus G_2$,
then join all adjacencies in $G_2 \setminus G_1$. But if we want to explore the solution
space, we have to observe that if an
adjacency $(a,b)$ exists in $G_1 \setminus G_2$ and an adjacency
$(a,c)$ exists in $G_2 \setminus G_1$, then first adjacency $(a,b)$
must be cut to create telomere $(a)$, and then telomere $(a)$ can be
connected to telomere $(c)$. Similarly, if extremity $c$ belongs to an
adjacency in $G_1\setminus G_2$, then it must be also cut before
connecting the two telomeres. Therefore there are restrictions on the
order of cuts and joins.

The allowed order of cuts and joins can be
read from the adjacency graph: When an $\SCJ$ operation acts on $G_1$ and thus
creates $G_1'$, it also acts
on the adjacency graph of $G_1$ and $G_2$ by transforming it
into the adjacency graph of $G_1'$ and $G_2$. Therefore
the transformation of $G_1$ into $G_2$ can be seen as a
transformation of the adjacency graph into trivial components.
We say that an $\SCJ$ scenario {\em sorts} the adjacency graph
if it transforms it into trivial components. As any $\SCJ$
operation in a most parsimonious scenario acts on a single component, we say that the
set of $\SCJ$ operations acting on that component {\em sort}
it if they transform it into trivial components.

We first give the way of computing the number of scenarios
for sorting one component. Then the
number of scenarios for several components will be deduced
by a combination of scenarios from each component.

Let $W(i)$ (respectively $M(i)$, $O(i)$ and $C(i)$) denote the number of most parsimonious $\SCJ$
scenarios sorting a $W$-shaped path (respectively $M$-shaped path, odd
path, cycle) with $i$ adjacencies in $G_1$.~The following dynamic
programming algorithm allows to compute all these numbers.

For a trivial component, no $\SCJ$ operation is needed so there is only one solution: the empty sequence. This gives
\begin{eqnarray}
C(1) & = & 1 \\
O(0) & = & 1
\end{eqnarray}

The smallest $W$-shaped path has 0 adjacency in $G_1$ and one in $G_2$. There is a unique
solution sorting it: add the adjacency. This gives
\begin{equation}
W(0) = 1
\end{equation}

A scenario of any other component starts with cutting an adjacency in $G_1$.
For a $W$-shaped path, this results in two $W$-shaped paths.
For an $M$-shaped path, this results in two odd paths.
For an odd path, this results in an odd path and a $W$-shaped path.
For a cycle, this results in a $W$-shaped path.
Each emerging component has fewer adjacencies in $G_1$,
and hence, a dynamic programming recursion can be applied: the
resulting components must be sorted and in case of two resulting
components, the sorting steps on the components must be merged. Hence
the dynamic programming recursions are
\begin{eqnarray}
C(i) &=& i \times W(i-1) \\
W(i) &=& \sum_{j=1}^i {2i \choose 2j-1} W(j-1) W(i-j)\\
M(i) &=& \sum_{j=1}^i {2i-2 \choose 2j-2} O(j-1) O(i-j) \\
O(i) &=& \sum_{j=1}^i {2i-1 \choose 2j-2} O(j-1) W(i-j)
\end{eqnarray}
These dynamic programming recursions can be used for counting and
sampling by the classical
Forward-Backward phases: in the Forward phase the number of
solutions is calculated, and in the Backward phase
one random solution is chosen based on the numbers in the sums.

So it is possible to compute $W(i)$, $M(i)$, $O(i)$ and $C(i)$ in polynomial
time and to sample one scenario from the uniform distribution.
We can then count and sample for several components by adding
a multinomial coefficient.

\begin{theorem}\label{theo:sharpMPSCJ}
Let $G_1$ and $G_2$ be two genomes with adjacency graph $AG$.
Assume $AG$ contains $i$ $M$-shaped paths,
with respectively $m_1, m_2, \ldots, m_i$ adjacencies in $G_1$;
$AG$ contains $j$ $W$-shaped paths,
with respectively $w_1,w_2,\ldots, w_j$ adjacencies in $G_1$;
$AG$ contains $k$ odd paths,
with respectively $v_1, v_2, \ldots, v_k$ adjacencies in $G_1$;
and $AG$ contains $l$ cycles,
with respectively $c_1,c_2, \ldots, c_l$ adjacencies in $G_1$.
The number of most parsimonious $\SCJ$ scenarios from $G_1$ to $G_2$ is
\begin{eqnarray}
\frac{\left(  \sum_{n=1}^i (2 m_n - 1) + \sum_{n=1}^j (2w_n+1) + \sum_{n=1}^k
(2v_n) + \sum_{n=1}^l (2c_n))\right)!}{\prod_{n=1}^i (2 m_n - 1)! \prod_{n=1}^j (2w_n+1)! \prod_{n=1}^k
(2v_n)! \prod_{n=1}^l (2c_n)!} &\times& \nonumber\\
\times \prod_{n=1}^i M(n) \prod_{n=1}^j W(n) \prod_{n=1}^k
O(n) \prod_{n=1}^l C(n) && \label{eq:numberofSCJ}
\end{eqnarray}
\end{theorem}

Sampling a scenario from the uniform distribution is then achieved by
generating a random
permutation with different colours and indices, one colour for each
component, and then wipe down the indices so get a permutation with
repeats. For each component, its sorting steps must be put into the
joint scenario indicated by the colour of the component.

We can then state the following theorem settling the complexity
of the comparison of two genomes by $\SCJ$.

\begin{theorem}\label{theo:sampling}
$\sharpMPSCJ$ is in $\FP$ and there is a polynomial algorithm
sampling from the exact uniform distribution of the solution space of an $\MPSCJ$ problem.
\end{theorem}

\subsection{Alternating permutations}

The solutions to $\sharpMPSCJ$ for single components are also linked to the
number of alternating permutations, for which finding a formula is an
old open problem. 
An {\em alternating permutation} of size $n$ is a permutation
$c_1,\dots,c_n$ of $\{1,\dots,n\}$ such that $c_{2i-1}<c_{2i}$ and
$c_{2i}>c_{2i+1}$ for all $i$ \citep{Andre1881}.
For example, if $n=4$, the permutation $1,3,2,4$ is an alternating permutation but
$1,3,4,2$ is not because 3 is less than 4.
The number of alternating permutations of size $n$ is denoted by $A_n$ and
finding these numbers is known as Andr\'e's problem.

We show that computing $\SCJ$ scenarios is closely related:

\begin{theorem}\label{lem:m}
$$M(k) = A_{2k-1}$$ 
$$W(k) = A_{2k+1}$$ 
$$O(k) = A_{2k}$$
$$C(k) = k\times A_{2k-1}$$
\end{theorem}
\begin{proof}
We prove only the first line, the second and the third lines can be
proved the same way. The proof of the last line comes from the fact
that a cycle with $k$ adjacencies can be opened in $k$ different ways
into a W-shaped component with $k-1$ adjacencies.
Let the adjacencies in the $G_1$ part of the $M$-shaped component be
$(x_1,x_2)$, $(x_3,x_4)$, $\ldots (x_{2k-1},x_{2k})$. Any $\SCJ$
scenario sorting these must cut all these adjacencies and must create
adjacencies $(x_2,x_3)$, $(x_4,x_5)$, $\ldots (x_{2k-2},x_{2k-1})$. Let
us index the $\SCJ$ operations in a scenario, and let $\pi_{2i-1}$ be the
index of the $\SCJ$ step which cuts the adjacency $(x_{2i-1},x_{2i})$, and
let $\pi_{2i}$ be the index of the $\SCJ$ step which joins
$x_{2i}$ and $x_{2i+1}$.

In any most parsimonious $\SCJ$ sorting the $M$-shaped component,
$\pi_{2i-1}<\pi_{2i}$ and $\pi_{2i+1}<\pi_{2i}$, so $\pi$ is an alternating
permutation. Hence the number of sorting scenarios is at most $A_{2k-1}$.

On the other hand, for any alternating permutation of size $2k-1$, we can
construct a sorting scenario in which the indexes come from the
alternating permutation. Since the sorting scenarios for different
alternating permutations are different, the number of $\SCJ$ scenarios is
at least $A_{2k-1}$. 
\end{proof}

%


\section{Counting and sampling $\SCJ$ small parsimony solutions}
\label{sec:arbitrarytree}

The $\SPSCJ$ problem is in $\P$, since one optimal assignment of
genomes to the internal nodes can be drawn in polynomial running time,
 \citep{fm2011}. However, we show that estimating the size of the solution space,
 as well as uniformly sampling it, is hard.

We show first that there is no polynomial
running time algorithm which samples almost uniformly from the
solutions unless $\RP = \NP$:
\begin{theorem}\label{theo:non-approx}
$\sharpSPSCJ \in \FPAUS \Rightarrow \RP = \NP$.
\end{theorem}

Then our conjecture is that $\sharpSPSCJ \in \sharpPcomp$, but we can prove only a
slightly weaker result
\begin{theorem}\label{theo:pnp}
$\sharpSPSCJ \in \FP \Rightarrow \P = \NP$.
\end{theorem}

Stochastic counting ($\FPRAS$) and sampling ($\FPAUS$) are equivalent
for self-reducible problems [\citep{Jerrum1986}, see the quite
technical definition of self-reducibility there]. However the counting
counterpart of Theorem ~\ref{theo:non-approx} cannot be immediately
deduced from it because we miss a proof of self-reducibility for
$\sharpSPSCJ$, which seems far from trivial, even not true in that
case. So we have to prove this counting counterpart independently. 

The construction we use in the proof of Theorem~\ref{theo:non-approx}
shows the hardness of a more specific problem and can be adapted to prove that:

\begin{theorem}\label{theo:hard_FPRAS}
$\sharpSPSCJ \in \FPRAS \Rightarrow \RP = \NP$.
\end{theorem}

We first recall in the following subsection how to draw one particular solution and then how to build
all possible solutions. Then we show how to generate an $\RP$ algorithm
for $\SAT$ using an $\FPAUS$ algorithm for the $\sharpSPSCJ$
problem. Since $\SAT \in \NPcomp$, this construction proves
Theorem~\ref{theo:non-approx}. This section finishes with proving
Theorems~\ref{theo:pnp}~and~\ref{theo:hard_FPRAS}. 

\subsection{The Fitch and Sankoff solutions}
\label{subsec:fitchsankoff}
Let $\Pi_1, \Pi_2, \ldots, \Pi_k$ denote the adjacency sets of genomes
$G_1, G_2, \ldots G_k$ and let
\begin{equation}
\Pi = \cup_{i=1}^k \Pi_i
\end{equation}

\cite{fm2011} proved that the parsimony score is equal to
the sum of the scores for each particular adjacency $\alpha\in\Pi$. This
can be computed by solving the small parsimony
problem for a discrete character. Although this is mainly textbook
material, we recall the principles of the standard algorithms solving
this problem for one adjacency because some stages will be referred to
in the hardness proof.  
%
%
For one adjacency, the small parsimony problem is solved by Fitch's
algorithm \citep{Fitch1971}. 
Its principle is first to assign sets
($\{0\}$, $\{1\}$ or $\{0,1\}$) to every node of the tree, visiting
the nodes of the tree in post-order traversal, ie. first the leaves of
the tree and then the parents of each node. At the leaves of the tree,
$\{0\}$s and $\{1\}$s are assigned according to the pattern of
presence or absence of $\alpha$ in the corresponding genomes. 
Let $B(\alpha, u)$ denote the set assigned to node $u$ regarding
adjacency $\alpha$. Fitch's algorithm applies the recursion 
\begin{equation}
B(\alpha,u) = \begin{cases} B(\alpha, v_1) \cap B(\alpha, v_2)
  &\mbox{if } B(\alpha, v_1) \cap B(\alpha, v_2) \neq \emptyset \\ 
B(\alpha, v_1) \cup B(\alpha, v_2) & \mbox{otherwise} \end{cases}
\end{equation}
where $v_1$ and $v_2$ are the children of $u$. 

\begin{definition}
We say that there is an {\em ambiguity} for an adjacency $\alpha$ at vertex $u$
if $B(\alpha,u) = \{0,1\}$.
\end{definition}

Then starting from the root, the nodes are visited in a pre-order
traversal, and $\{0\}$ or $\{1\}$ is assigned to each node according to
the following rules: If $B(\alpha, \root)$
contains only one element, then it is assigned to the root.
If $B(\alpha, root) = \{0,1\}$, then any of them
can be chosen for the root. Once the number assigned to the root is
fixed, the values are propagated down. Let $F(\alpha, v)$ denote
the singleton set assigned to the node $v$ for adjacency
$\alpha$. Fitch's algorithm applies the recursion:
\begin{equation}
F(\alpha, v) =  \begin{cases} F(\alpha, u) \cap B(\alpha, v)
  &\mbox{if } F(\alpha, u) \cap B(\alpha, v) \neq \emptyset \\ 
B(\alpha, v) & \mbox{otherwise} \end{cases}
\end{equation}
where $v$ is a child of $u$.

$F(\alpha,v)$ then always contains exactly one element. Doing this independently
for all adjacencies does not guarantee that the collection of present adjacencies
at each node is a genome: we call a subset $\Sigma \subseteq \Pi$  a \emph{valid genome} if
there is no couple of adjacencies $\alpha_1,\alpha_2 \in \Sigma$ with a common
extremity. \cite{fm2011} showed that if the assignments of $F(\alpha, root)$ over all
possible adjacencies $\alpha \in \Pi$ are chosen to be a valid genome,
then all genomes at the internal nodes are also valid (deduced from Lemmas
6.1. and 6.2 in \cite{fm2011}). They also proved that at least one
valid assignment exists since the Fitch's algorithm never gives
non-ambiguous values for adjacencies sharing extremities.

We call {\em Fitch solutions} the genome assignments constructed this way. However, they are not
the only possible most parsimonious genome assignments. Some of them cannot be found
by Fitch's algorithm.
All solutions can be found by a generalization of Fitch's algorithm, Sankoff's algorithm \citep{sr1975}.
It is a dynamic programming principle which computes two values for each node of the phylogenetic tree:
for a leaf $v_i$ assigned with genome $G_i$,
\begin{equation}
s1(\alpha, v_i) = \begin{cases} 0
  &\mbox{if } \alpha \in \Pi_i \\ 
\infty & \mbox{otherwise} \end{cases}
\end{equation}
\begin{equation}
s0(\alpha, v_i) = \begin{cases} 0
  &\mbox{if } \alpha \notin \Pi_i \\ 
\infty & \mbox{otherwise} \end{cases}
\end{equation}
and for an internal node $u$ with children $v_1$ and $v_2$:
\begin{eqnarray}
s1(\alpha, u) & = & \min\{s1(\alpha,v_1), s0(\alpha, v_1) + 1\} +
\nonumber \\ && \min\{s1(\alpha,v_2), s0(\alpha, v_2) + 1\}\label{eq:sankoff1}  \\
s0(\alpha, u) & = & \min\{s0(\alpha,v_1), s1(\alpha, v_1) + 1\} +
\nonumber \\
&& \min\{s0(\alpha,v_2), s1(\alpha, v_2) + 1\} \label{eq:sankoff2}
\end{eqnarray}
The value of $s0(\alpha,u)$ (respectively $s1(\alpha,u)$) represents the minimum number of edges
under the subtree rooted at $u$ which are labelled with different
presence/absence of $\alpha$ at their two ends in a most parsimonious
scenario, given that $u$ is labelled with the absence (respectively presence) of
$\alpha$.
Then $\min(s1(\alpha,root),s0(\alpha,root))$ is the minimum small parsimony
solution for adjacency $\alpha$, and the assignments to internal nodes are obtained by
propagating down the values based on which gave the minimum in
Equations~\ref{eq:sankoff1}~and~\ref{eq:sankoff2}.

Contrary to Fitch's algorithm, this one explores all possible most
parsimonious assignments for a given adjacency
\citep{esz1994}. Unfortunately, in that case there is no guarantee that
all of these assignments give valid genomes, as Feijao and Meidanis
result holds only for Fitch's solutions. 
%
It is an open question how to estimate the number of most parsimonious genome assignments (we
can call them the {\em Sankoff solutions}), and is beyond the scope of this paper (note that
it is a different problem from $\sharpSPSCJ$ where we aim at estimating the number of
$\SCJ$ scenarios and not only genome assignments).




\subsection{Sampling most parsimonious $\SPSCJ$ scenarios is hard}
\label{subsec:hardness}

In this section we construct a problem instance $x \in \SPSCJ$ for any
$\CNF$ formula $\Phi$ with $n$ variables, such that 
if there exists an $\FPAUS$ for $x$ then it is an $\RP$ algorithm for
deciding whether or not $\Phi$ is satisfiable.

Let $\Phi$ be a $\CNF$ with $n$ logical variables and $k$ clauses. We
are going to construct a tree denoted by $T_\Phi$, and label its
leaves with genomes. For each logical
variable $b_i$ we create an adjacency $\alpha_i$. In this construction, all
adjacencies are independent one from another, namely they never
share common extremities. So there is no genome validity issue in this
construction, any assignment of adjacency presence/absence is a valid genome.

For each clause $c_j$, we construct a subtree $T_{c_j}$. The construction is
done in three phases, see also
Figure~\ref{fig:Tc}. First, we create a constant size subtree, called
unit subtree using building blocks we call elementary subtrees. Then
in the blowing up phase, this unit subtree is repeated several times,
and in the third phase it is amended with another constant size
subtree. The reason for this construction is the following: the unit
subtree is constructed in such a way that if a clause is satisfied,
the number of $\SCJ$ solutions is a greater number, and is always the
same number not depending on how many literals provide satisfaction of
the clause. When the clause is not satisfied, the number of $\SCJ$ 
solutions is a smaller number. The blowing up is necessary for
sufficiently separating the number of solutions for satisfying and
not satisfying assignments. Finally, the amending is necessary for
having all adjacencies ambiguous in the Fitch solutions.

\begin{figure}[tpb]
\centerline{\includegraphics[scale=0.6]{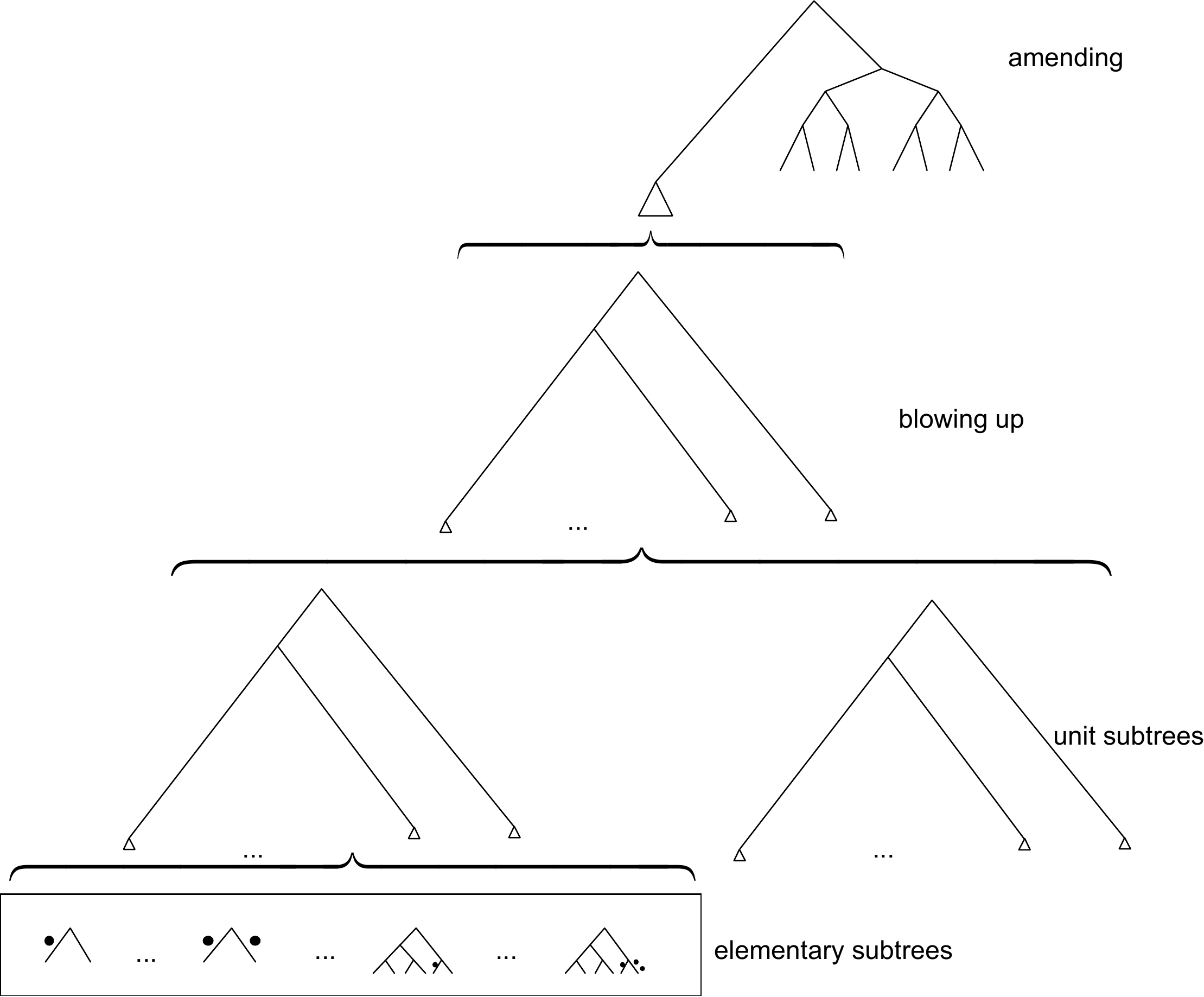}}
\caption{Constructing a subtree $T_{c_j}$ for a clause $c_j$. The subtree
  is built in three phases. First, elementary subtrees are connected
  with a comb to get a unit subtree. In the second phase the same unit
  subtree is repeated several times 'blowing up' the tree. In the
  third phase, the blown up tree is amended with a constant size,
  depth 3 fully balanced tree. The smaller subtrees constructed in the
  previous phase are denoted with a triangle in the next phase. See also text for details.}\label{fig:Tc}
\end{figure}

We detail the construction of the subtree for the clause
$c_j = b_1 \vee b_2\vee b_3$, denoted by
$T_{c_j}$. Subtrees for the other kinds of clauses are 
constructed similarly. 
The unit subtree is built from $76$ smaller subtrees that we will call
\emph{elementary subtrees}. Only $14$ different types of elementary
subtrees are in a unit subtree, but several of them have given
multiplicity, and the total count of them is $76$, see also
Table~\ref{tab:solutions}. Some of the elementary
subtrees are cherry motives for which we arbitrarily identify a left
and a right leaf. On some of these cherries, we add one or more
adjacencies, called extra adjacencies, 
which are present exactly on one leaf of the cherry and absent
everywhere else in $T_\Phi$.
So the edges connecting these leaves to the rest of the entire tree $T_\Phi$ will
contain one or more additional $\SCJ$ operations in all
most parsimonious solutions.

A clause contains $3$
logical variables, the unit subtree will be such that for the corresponding
adjacencies, Fitch's algorithm assigns an ambiguity at the root of the
subtree $T_{c_j}$, namely
\begin{equation}
B(\alpha_i,\root) = \{0,1\}
\end{equation}
for each $b_i \in c_j$. The entire tree, $T_\Phi$, will also be such that Sankoff solutions
are all found by Fitch's algorithm, namely, all solutions can be found
by the Fitch's algorithm, as we are going to state and prove in
Lemma~\ref{lem:onlyFitch}. Therefore there will be $8$ possible
genome assignments for
the unit subtree, related to the $8$ possible assignments of the three logical
variables at the root. Let the presence of the adjacency at the root mean logical true value,
and let absence mean logical false value. The constructed unit subtree will
be such that if the clause is not
satisfied, the number of possible $\SCJ$ scenarios for the corresponding
assignment on this unit subtree is $2^{136} \times 3^{76}$, and if the clause is
satisfied, then the number of possible $\SCJ$ scenarios for each
corresponding assignment is $2^{156} \times 3^{64}$. The ratio of the two
numbers is $2^{20}/3^{12} > 1$. We will denote this number by $\gamma$.
This ratio will be the basis for our proof: any $\FPAUS$
will sample the solutions corresponding to the satisfied clauses more often than the
non-satisfied ones because the former are more numerous. This
can be turned into an $\RP$ algorithm for $\SAT$.

Below we detail the construction of the elementary subtrees and also
give the number of $\SCJ$ solutions on them since the number of
solutions on the unit subtree is simply the product of these numbers.

For the adjacencies $\alpha_1$, $\alpha_2$ and $\alpha_3$, the
cherries are the following:
\begin{itemize}
\item for the cherries on which the left leaf contain one extra
  adjacency, the presence/absence pattern on the left and right leaf is given by\\

\noindent $011$, $100$ \\
\noindent $101$, $010$ \\
\noindent $110$, $001$ \\
\noindent $000$, $111$ \\

The first column shows the presence/absence of the three adjacencies
on the left leaf, the second column shows the presence/absence of the
three adjacencies on the right leaf. Hence, for example,
$000$ means that none of the adjacencies is present, $100$ means
that only the first adjacency is present.
The number of $\SCJ$ solutions on one cherry is $24$ if the assignment of
adjacencies at the root of the cherry is the same as on the right
leaf. Indeed, in that case, $4$ $\SCJ$ operations are necessary on the
left edge, and they can be performed in any order. If the number
of $\SCJ$ operations are $3$ and $1$ respectively on the left and
right edges, or \emph{vica versa}
the number of solutions is $6$. Finally, if both edges have $2$ $\SCJ$ operations,
then the number of solutions is $4$.

\item There is one cherry without any extra adjacency, and its
  presence/absence pattern is\\

\noindent $000$, $111$\\

If the clause is not satisfied, the number of $\SCJ$ solutions on this
cherry is $6$; if all logical values are true, the number of $\SCJ$ solutions is still $6$;
in any other case, the number of $\SCJ$ solutions is $2$. 

This elementary subtree is repeated 3 times.

\item Finally, there are cherries with one-one extra adjacency on both
  leaves. These are two different adjacencies, so both of them need
  one extra $\SCJ$ operation on their incoming edge. The
  presence/absence patterns are\\

\noindent $011$, $100$\\
\noindent $101$, $010$\\
\noindent $110$, $001$\\

If all $\SCJ$ operations due to $\alpha_i$, $i=1,2,3$ falls onto one edge,
then the number of solutions is $24$, otherwise the number of
solutions is $12$.

Each of these elementary subtrees are repeated $15$ times.

\end{itemize}

The remaining elementary subtrees contain $3$ cherries connected
with a comb, that is, a completely unbalanced tree, see also Figure \ref{fig:3x2tree}.
For the cherry at the right end of this elementary subtree, we add one or more adjacencies that are
present on one of the leaves and absent everywhere else in $T_\Phi$. When there is
one extra adjacency on the left leaf, the adjacencies
$\alpha_1$, $\alpha_2$ and $\alpha_3$ are assigned with the following
presences/absences on the three cherries at the top of the three combs: \\

\noindent $011$, $000$ \\
\noindent $101$, $000$ \\
\noindent $110$, $000$ \\

Again, the first column shows the assignment for the left leaf, the second column
for the right leaf. The number of $\SCJ$ solutions is $6$ on this cherry if the assignment
at the root is $0$ for both adjacencies which has assignment $1$ on the
left leaf. In any other cases, the number of solutions is $2$. Two of
the adjacencies are ambiguous on this cherry, and the third one is~$0$.
On the remaining two cherries of this elementary subtree, this third adjacency
is present on all leaves, while the other two are made ambiguous in such a way that
any assignment has one $\SCJ$ scenario on the remaining of the tree. We show the solution for the
first subtree on Figure~\ref{fig:3x2tree}.

Each of these elementary subtrees are repeated $3$ times.

Finally, there are elementary subtrees when there is $1$ extra adjacency on the left
leaf and $2$ extra adjacencies on the right leaf. The assignments are\\

\noindent $011$, $000$ \\
\noindent $101$, $000$ \\
\noindent $110$, $000$ \\

The number of $\SCJ$ solutions is $24$ on this cherry if both
necessary $\SCJ$ operations fall onto the edge having $2$ additional
$\SCJ$ operations
due to the extra adjacencies, and $12$ in all other cases.

Each of these elementary subtrees are repeated $5$ times.

\begin{figure}[tpb]
\centerline{\includegraphics[scale=0.6]{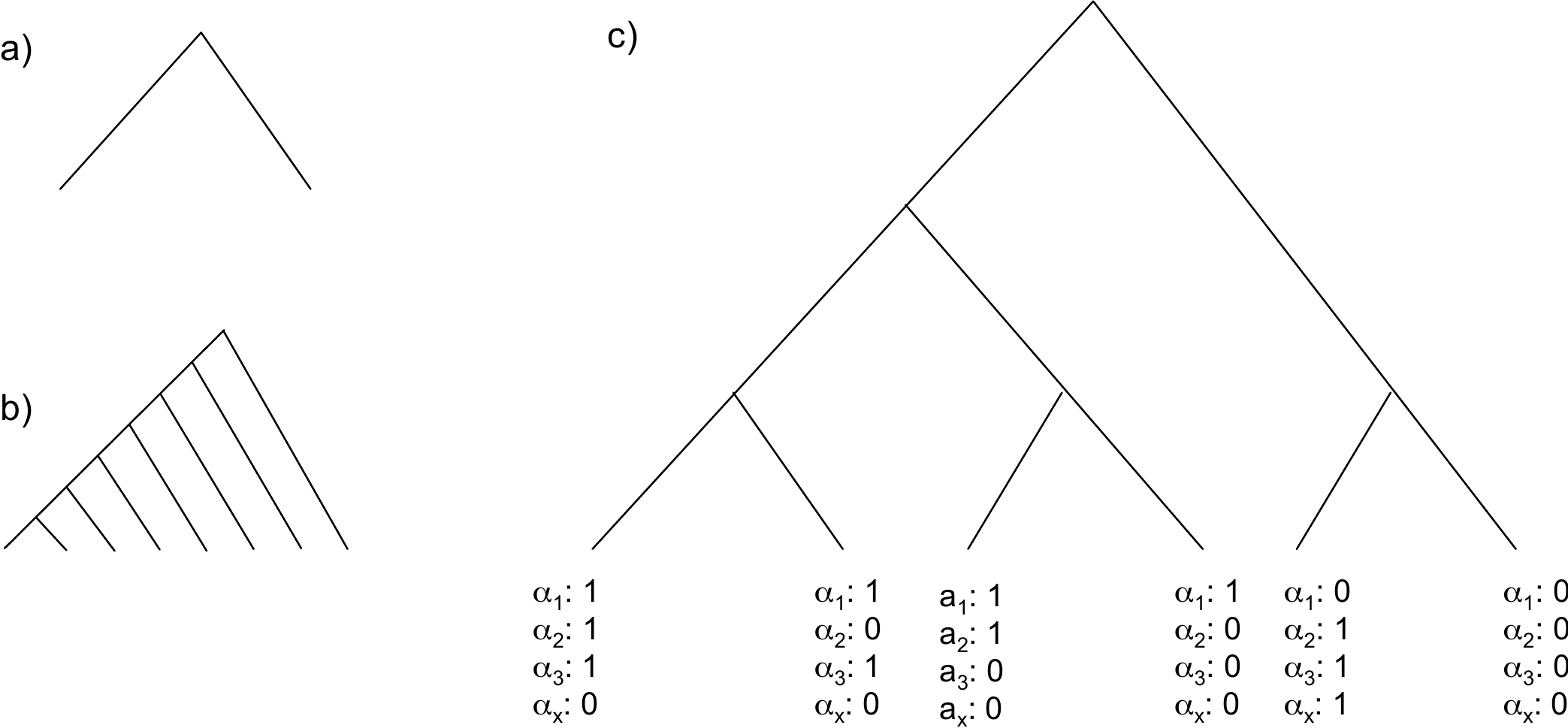}}
\caption{{\bf a)} A cherry motif, ie., two leaves connected with an
  internal node. {\bf b)} A comb, ie., a fully unbalanced tree with $8$
  leaves. {\bf c)} A tree with $3$ cherry motifs connected with a comb. The
  assignments for $4$ adjacencies, $\alpha_1$, $\alpha_2$, $\alpha_3$ and $\alpha_x$ are
  shown at the bottom for each leaf. $\alpha_i$, $i=1,2,3$ are the
  adjacencies related to the logical variables $b_i$, and $\alpha_x$ is an
  extra adjacency. Note that Fitch's algorithm gives ambiguity for all
  adjacencies $\alpha_i$ at the root of this subtree.}\label{fig:3x2tree}
\end{figure}

\begin{table}
{\tiny
\begin{tabular}{c||c|c|c|c|c|c|c|c|c|c|c|c|c|c}
        ~    & \includegraphics[scale=0.1]{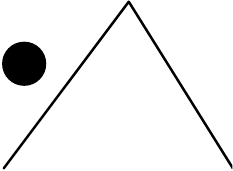} & 
                  \includegraphics[scale=0.1]{cherryandscj.pdf} &
                  \includegraphics[scale=0.1]{cherryandscj.pdf} &
                  \includegraphics[scale=0.1]{cherryandscj.pdf}  &
                   \includegraphics[scale=0.1]{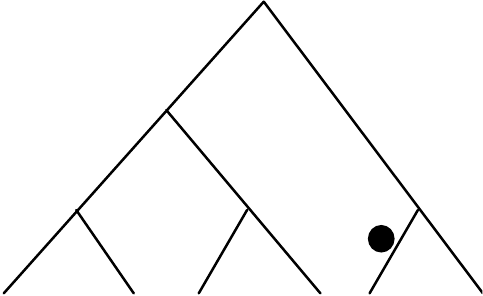}  & 
                   \includegraphics[scale=0.1]{comb.pdf}  & 
                   \includegraphics[scale=0.1]{comb.pdf}  & 
                  \includegraphics[scale=0.1]{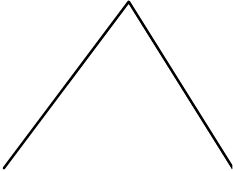}  & 
                   \includegraphics[scale=0.1]{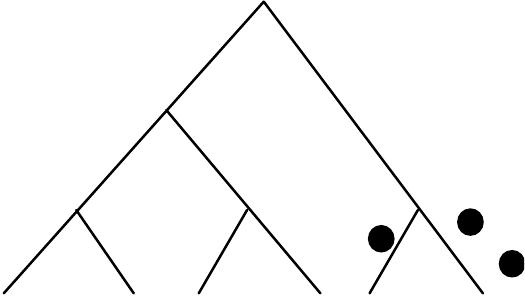}  & 
                   \includegraphics[scale=0.1]{3x2tree3dots.pdf}  &
                   \includegraphics[scale=0.1]{3x2tree3dots.pdf}  &
                  \includegraphics[scale=0.1]{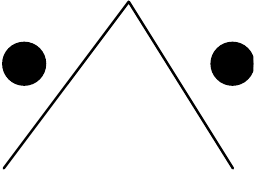}  &
                  \includegraphics[scale=0.1]{cherry2dots.pdf}  &
                  \includegraphics[scale=0.1]{cherry2dots.pdf}  \\

        ~    & 011 & 101 & 110 & 000 & 011 & 101 & 110 & 000 & 011 & 101 & 110 & 011 & 101 & 110\\ \hline
       \#   & 1 & 1 & 1 & 1 & 3 & 3 & 3 & 3 & 5 & 5 & 5 & 15 &15 &15\\ \hline
        000 &  6  &   6  &   6 &   6 & $6^{3}$ & $6^{3}$ & $6^{3}$ & $6^{3}$ & $12^{5}$ & $12^{5}$ & $12^{5}$ & $12^{15}$ & $12^{15}$ & $12^{15}$ \\  
        100 & 24 &   4  &   4 &   4 & $6^{3}$ & $2^{3}$ & $2^{3}$ & $2^{3}$ & $12^{5}$ & $12^{5}$ & $12^{5}$ & $24^{15}$ & $12^{15}$ & $12^{15}$ \\  
        010 &   4 & 24  &   4 &   4 & $2^{3}$ & $6^{3}$ & $2^{3}$ & $2^{3}$ & $12^{5}$ & $12^{5}$ & $12^{5}$ & $12^{15}$ & $24^{15}$ & $12^{15}$ \\  
        110 &   6 &   6  &   6 &   6 & $2^{3}$ & $2^{3}$ & $2^{3}$ & $2^{3}$ & $12^{5}$ & $12^{5}$ & $24^{5}$ & $12^{15}$ & $12^{15}$ & $24^{15}$ \\  
        001 &   4 &   4  & 24 &   4 & $2^{3}$ & $2^{3}$ & $6^{3}$ & $2^{3}$ & $12^{5}$ & $12^{5}$ & $12^{5}$ & $12^{15}$ & $12^{15}$ & $24^{15}$ \\  
        101 &   6 &   6  &   6 &   6 & $2^{3}$ & $2^{3}$ & $2^{3}$ & $2^{3}$ & $12^{5}$ & $24^{5}$ & $12^{5}$ & $12^{15}$ & $24^{15}$ & $12^{15}$ \\  
        011 &   6 &   6  &   6 &   6 & $2^{3}$ & $2^{3}$ & $2^{3}$ & $2^{3}$ & $24^{5}$ & $12^{5}$ & $12^{5}$ & $24^{15}$ & $12^{15}$ & $12^{15}$ \\  
        111 &   4 &   4  &   4 & 24 & $2^{3}$ & $2^{3}$ & $2^{3}$ & $6^{3}$ & $24^{5}$ & $24^{5}$ & $24^{5}$ & $12^{15}$ & $12^{15}$ & $12^{15}$  
\end{tabular}
}
\caption{The number of $\SCJ$ scenarios on different
  elementary subtrees of the unit subtree of the
subtree $T_{c_j}$ for clause $c_j = b_1 \vee b_2 \vee b_3$. Columns represent the $14$
different types of components, the topology of the elementary subtree is indicated
on the top. The black dot means extra $\SCJ$ operations on the indicated edge, the
numbers represent the presence/absence of adjacencies on the left leaf
of a particular cherry, see text for details. The row starting with \# indicates the
number of repeats of the elementary subtrees.  Further rows represent the logical true/false
values of $b_i$s, for example, $001$ means $b_1 = false$, $b_2=false$,
$b_3=true$. The values in the table indicate the number of solutions,
raised to the appropriate power due to multiplicity of the
elementary subtrees. It is easy to check that the product of the numbers in the 
first line is $2^{136} \times 3^{76}$ and in any other lines is $2^{156}  \times
3^{64}$.}
\label{tab:solutions}
\end{table}

In this way, the roots of all 76 elementary subtrees are ambiguous for the three
adjacencies related to logical variables. We connect the 76 elementary
subtrees with a comb, and thus, all three
adjacencies are ambiguous at the root of the entire subtree, which is
the unit subtree. If the
clause is satisfied, the number of SCJ scenarios for the corresponding
assignment is $2^{156}  \times 3^{64}$, if the clause is not satisfied, the
number of SCJ solutions is $2^{136} \times 3^{76}$, as can be checked
on Table~\ref{tab:solutions}. The ratio of them is indeed 
$2^{20}/3^{12} = \gamma$. The number of leaves on this unit subtree is
$248$, and $148$ additional adjacencies are introduced.

This was the construction of the constant size unit subtree.
In the next step, we ``blow up'' the system. Similar blowing up can be
found in \cite{Jerrum1986}, in the proof of Theorem 5.1. We repeat the above described
unit subtree $\left\lceil(k \log((n-3)!) + n \log(2)) / \log(\gamma)\right\rceil +1$ times, and 
connect all of them with a comb (completely unbalanced tree). All three adjacencies representing the three
logical variables in the clause are still ambiguous at the root of
this blown up subtree, and thus, there are still $8$ Fitch solutions.
For a solution satisfying the clause, the number of $\SCJ$ scenarios
on this blown up subtree
is
\begin{equation}
X = \left(2^{156} \times 3^{64}\right)^{\left\lceil\frac{k \log((n-3)!) + n
    \log(2)}{\log(\gamma)}\right\rceil +1}
\end{equation}
and the number of scenarios if the clause is not satisfied is
\begin{equation}
Y = \left(2^{136} \times 3^{76}\right)^{\left\lceil\frac{k \log((n-3)!) + n
    \log(2)}{\log(\gamma)}\right\rceil +1}
\end{equation}
Let all adjacencies
not participating in the clause be $0$ on this blown up subtree.

We are close to the final subtree $T_{c_j}$ for one clause, $c_j$. In the
third phase, we amend the so far obtained tree with a constant size
subtree. Construct a fully balanced 
depth $3$ binary tree, on which all $3$
adjacencies which are in the clause are ambiguous at the root without making more
than $1$ $\SCJ$ scenario on it, similarly to the left part of the tree
on Figure~\ref{fig:3x2tree}. All other adjacencies not participating in
the clause are present at all leaves of this tree.

Here is how to construct $T_{c_j}$ for one clause, $c_j$. Construct an additional vertex which
will be its root. The left child of the root is the blown up tree, while its right child is the depth $3$
balanced tree. Denote by $T_{c_j}$ this final tree for one clause~$c_j$.

All adjacencies are ambiguous at the root of the subtree $T_{c_j}$, therefore there are $2^n$ Fitch
solutions for the assignments of the internal nodes of $T_{c_j}$.

\begin{lemma}
For any assignment of the $n$ adjacencies, if the clause $c_i$ is
satisfied, then the number of $\SCJ$ scenarios for the corresponding
assignment on $T_{c_j}$ is at least 
\begin{equation}
Y \times ((n-3)!)^k \times 2^n \times \gamma
\end{equation}
and at most
\begin{equation}
Y \times ((n-3)!)^{k+1} \times 2^n \times \gamma
\end{equation}
If the clause is not satisfied, then the number of
$\SCJ$ scenarios is at most $Y \times (n-3)!$.
\end{lemma}
\begin{proof}
The $B$ values of Fitch's algorithm for the $n-3$ adjacencies not
representing a logical value in the clause $c_i$ are all $\{0\}$ at all
the nodes of
the left child of the root, and all $\{1\}$ at all the nodes of the right child of the
root. Therefore in all scenarios there are $n-3$ cumulated $\SCJ$ operations on the two edges
going out of the root. If they are all on one of the edges,
the number of possible $\SCJ$ scenarios is $(n-3)!$, and in all other
cases they are less, but at least $1$. (Actually, the minimum is
$\left(((n-3)/2)!\right)^2$, but the very loose lower bound $1$ is
sufficent for our calculations).
Then if the clause is satisfied, the number of $\SCJ$ scenarios is
between $X$ and $X\times (n-3)!$. Note that
$$X/Y = ((n-3)!)^k\times 2^n\times \gamma,$$
which gives the stated result.
If the clause is not satisfied, the number of $\SCJ$ scenarios is at most $Y \times (n-3)!$.
\end{proof}

For all $k$ clauses, construct such a subtree and connect all of them
with a comb. This is the final tree $T_{\Phi}$ for the $\CNF$ $\Phi$.

All adjacencies corresponding to logical variables
are ambiguous at the root of the $T_\Phi$, so there are
$2^n$ Fitch solutions. We prove that there is the same number
of Sankoff solutions.

\begin{lemma}\label{lem:onlyFitch}
All adjacency assignments for the $\SPSCJ$ problem on tree $T_{\Phi}$
are Fitch solutions.
\end{lemma}
\begin{proof}
There are two types of adjacencies participating in $T_{\Phi}$.
There are $n$ of them related to the logical variables in $\Phi$, the other
adjacencies are introduced in the construction and are present on exactly one leaf,
absent everywhere else in $T_\Phi$.

If an adjacency $\alpha_x$ is present only on one leaf, then in any
$\SPSCJ$ solution it is created on the edge connecting the leaf to
the remaining part of the tree. This solution is provided by
Fitch's algorithm.

The tree is constructed in such way that for all $\alpha_i$
representing variable $b_i \in \Phi$,
\begin{equation}
B(\alpha_i,v) = \{0,1\} \Rightarrow B(\alpha_i,u) = \{0,1\}
\end{equation}
where $u$ is the parent of $v$. First observe that
\begin{equation}
B(\alpha_i,v) = \{0,1\} \Rightarrow s1(\alpha_i,u) = s0(\alpha_i,u)
\end{equation}
this means that whenever the two children $v_1$ and $v_2$ of
a node $u$ are ambiguous in Fitch's algorithm, 
\begin{eqnarray}
s1(\alpha_i,u) &= &s1(\alpha_i,v_1) + s1(\alpha_i,v_2) \\
s0(\alpha_i,u) &= &s0(\alpha_i,v_1) + s0(\alpha_i,v_2) 
\end{eqnarray}
namely, all Sankoff solutions are Fitch solutions.

Moreover, at any node $u$ where the $B$ value is ambiguous for some adjacency $\alpha_i$, 
while it is not ambiguous in the children of $u$, we have
\begin{equation}
s0(\alpha_i,u) = s1(\alpha_i,u) = 1
\end{equation}
and here again the Fitch solutions are the same as the
Sankoff solutions. 
\end{proof}

Now we are ready to prove Theorem~\ref{theo:non-approx}.

\begin{proof} {\bf (Theorem~\ref{theo:non-approx}.)}
Let $\Phi$ be a $\CNF$ with $k$ clauses. The number of Boolean
variables in $\Phi$ is at most $3k$, hence the tree $T_{\Phi}$
contains at most 
\begin{equation}
\left(248 \times \left(\left\lceil\frac{k \log((3k-3)!) + 3k \log(2)}{\log(\gamma)}\right\rceil +1\right) + 8\right)
\times k \label{eq:leaves}
\end{equation}
leaves, and 
\begin{equation}
6k + 296k\times \left(\left\lceil\frac{k \log((3k-3)!) + 3k
    \log(2)}{\log(\gamma)}\right\rceil +1\right) \label{eq:adj}
\end{equation}
extremities (twice the number of independent adjacencies
appearing). To explain Equation~\ref{eq:leaves}, $248$ is the number
of leaves on the unit subtree, it is reapeted $\left\lceil\frac{k \log((n-3)!) + 3k
\log(2)}{\log(\gamma)}\right\rceil +1$ times, an upper bound for $n$ is $3k$, as
mentioned above, and there are $8$ further leaves in the amending
phase of the construction of a subtree $T_{c_j}$ for a clause
$c_j$. Finally, there are $k$ clauses. To explain Equation~\ref{eq:adj},
there is an adjacency for each boolean variable, there are at most
$3k$ of them, each of them having $2$ extremities, yielding $6k$
extremities at most. There are $148$ extra adjacencies in each unit
subtree, having $296$ extremities. Each unit subtree is repeated
$\left\lceil\frac{k \log((n-3)!) + 3k
    \log(2)}{\log(\gamma)}\right\rceil +1$ times, upperly bounded by
$\left\lceil\frac{k \log((3k-3)!) + 3k
    \log(2)}{\log(\gamma)}\right\rceil +1$, and this is done for each
$k$ clauses. 
 
Hence the input size for the $\SPSCJ$ problem is a
polynomial function of the size of $\Phi$.

If $\Phi$ is satisfiable, then there exists an assignment for which
the number of $\SCJ$ scenarios is at least
\begin{equation}
Y^k \times (((n-3)!)^k \times 2^n \times \gamma)^k
\end{equation}
If at least one of the clauses is not satisfied, then the total number of $\SCJ$ scenarios
is at most
\begin{equation} 
Y^k \times (((n-3)!)^k \times 2^n \times \gamma)^{k-1} \times \left((n-3)!\right)^k
\end{equation} 
Therefore, if $\Phi$ is satisfiable, there are at most $2^n-1$ assignments which do not satisfy
the $\Phi$, and the number of corresponding $\SCJ$ scenarios is at
most 
\begin{equation}
Y^k \times (((n-3)!)^k  \times 2^n  \times \gamma)^{k-1} \times \left((n-3)!\right)^k \times (2^n-1).
\end{equation} Hence if $\Phi$ is satisfiable, then the number of
$\SCJ$ scenarios related to satisfying assignments are more than the
number of other $\SCJ$ scenarios. If an $\FPAUS$ exists for all most parsimonious
scenarios, then it would sample satisfying scenarios with more than
$0.5$ probability. Then this is an $\RP$ algorithm for $\SAT$. 
An $\RP$ algorithm for $\SAT$ immediately implies that $\RP = \NP$ \citep{Papadimitriou1993}.
\end{proof}

\subsection{Counting problems}

The same construction is sufficient to prove Theorem~\ref{theo:pnp}.
\begin{proof}{\bf (Theorem~\ref{theo:pnp}.)}
Assume that there is an $\FP$ algorithm for $\sharpSPSCJ$. Then for
any $\CNF$ $\Phi$, construct the above introduced problem instance
$x \in \sharpSPSCJ$, and calculate the exact number of solutions. If
$\Phi$ is not satisfiable, then the number of solutions is at most
\begin{equation}
Y^k \times (((n-3)!)^k  \times 2^n  \times \gamma)^{k-1} \times \left((n-3)!\right)^k \times 2^n\label{eq:notsat}
\end{equation}
If $\Phi$ can be satisfied, then the number of solutions is more than
\begin{equation}
Y^k \times (((n-3)!)^k  \times 2^n  \times \gamma)^k\label{eq:sat}
\end{equation}
Since the number in Equation~\ref{eq:sat} is greater than the number
in Equation~\ref{eq:notsat}, and the number of digits of these numbers
grows only polynomially with $|\Phi|$, given an $\FP$ algorithm for
$\sharpSPSCJ$, it would be decidable in polynomial running time
whether or not $\Phi$ is satisfiable. Since $\SAT \in \NPcomp$,
it would imply that $\P = \NP$.
\end{proof}

Now we prove the counting counterpart of the same result, that is,
$\sharpSPSCJ$ is not in $\FPRAS$ unless $\RP=\NP$. For this we need to define a
more restricted problem. 

\begin{definition}
The $\sharpFSPSCJ$ problem asks for the number of Fitch solutions of an
$\SPSCJ$ instance where pairs of adjacencies never share an extremity
and the values of a set of ambiguous adjacencies are fixed.
\end{definition}

Although \emph{stricto sensu}, the $\sharpFSPSCJ$ is still not a self
reducible counting problem, we can prove that it has an $\FPAUS$
algorithm if it has an $\FPRAS$ algorithm. Before proving it, we
discuss in a nutshell how to construct an $\FPAUS$ algorithm from an
$\FPRAS$ algorithm for self-reducible counting problems. The
description is not detailed, for a strict mathematical description,
see \cite{Sinclair1992}.

The heart of the method that creates an $\FPAUS$ from a self-reducible
counting problem in $\FPRAS$ is a rejection sampler
\citep{neumann1951}. A random solution is drawn sequentially
travelling down the counting tree of the self-reducible problem, using
the $\FPRAS$ approximations for the children of the current node, and
at each internal node the sampling probability is calculated. The
sampling probabilities are used to calculate the so called rejection
rate, the probability that the sample will be rejected. The central
theorem of the rejection method states that the accepted samples come
from sharp the uniform distribution. To transform this into an
$\FPAUS$, the rejection rate should be relatively small, so in a few
(polynomial number of) trials, the probability that all trials are
rejected becomes negligible. If all trials are rejected, then an
arbitrary solution is drawn, but due to its extremely small
probability, it causes a very small deviation from the uniform
distribution (measured in variational distance).

\begin{lemma}\label{lem:almost_self_reducible}
$\sharpFSPSCJ \in \FPRAS \Rightarrow \sharpFSPSCJ \in \FPAUS$
\end{lemma}
\begin{proof}
It is sufficient to show that the solutions can be put onto a counting
tree such that the depth of the tree is $O(poly(|x|))$  where $|x|$ is
the size of the problem instance, and for any internal node, one of the following is true:
\begin{itemize}
\item The number of descendants of the internal node is
$O(poly(|x|))$ where $|x|$ is the size of the problem instance, and for each descendant, a problem
$x' \in \sharpFSPSCJ$ exists whose number of solutions is the number of
leaves of that tree, and $|x'| = O(poly(|x|))$.
\item The number of descendants is $O(c^{poly(|x|)})$ for some $c>1$,
  but a perfect sampler exists that can sample sharp the uniform
  distribution of the descendants and the number of descendants can be
  calculated, both the sampler and the counter run in $O(poly(|x|))$
  time. Furthermore, all descendants are leaves.
\end{itemize} 
The algorithms in the second case provide that the protocol constructing an $\FPAUS$
sampler using an $\FPRAS$ algorithm described briefly above
can be done also for those nodes which have suprapolynomial number of
descendants but counting their number as well as sharp uniform sampling
them can be done in polynomial time. Indeed, both sampling and
calculating the sampling probabilities can be done in polynomial
running time, and it is easy to see that the strict uniform sampling
does not increase the rejection rate.

Fix an arbitrary total ordering of adjacencies.
Let $u$ be an internal node, and $z$ is the associated problem to
it. If there are ambiguities at the root of the evolutionary tree,
then take the smallest adjacency with ambiguity and without a
constraint, let it be denoted by $\alpha$. Then $u$ will have two
descendants, and they are associated with a problem instance where
problem instance $z$ is modified such that $\alpha$ has constraint
$0$ and constraint $1$. 

If $z$ does not have any ambiguity, then its assignment is unique. For
this unique assignment, the number of $\SCJ$ scenarios along each edge
can be counted and sharply uniformly sampled (Theorem~\ref{theo:sampling}), so these will be the
descendants of $u$ and also the leaves below $u$.
\end{proof}

The next lemma leads directly to the proof of
Theorem~\ref{theo:hard_FPRAS}.
\begin{lemma}\label{lem:2fpras}
$\sharpSPSCJ \in \FPRAS \Rightarrow \sharpFSPSCJ \in \FPRAS$
\end{lemma}
\begin{proof}
Let $x$ be a problem instance from $\sharpFSPSCJ$. Let $A$ denote the
set of adjacencies which are ambiguous at the root, but there are
constraints on them. Let $T$ denote the evolutionary tree of the
problem instance $x$.

We construct another problem $x'$, which has the same number of $\SCJ$
solutions but there are no ambiguities for those adjacencies which are
in $A$. We remove each $\alpha \in A$, and introduce new, independent
adjacencies. For any $\alpha \in A$, let $E(\alpha)$ denote the set of
edges of $T$ for which an SCJ operation is necessary with the
prescribed assignment of $\alpha$. We introduce $|E(\alpha)|$ new,
independent adjacencies in the following way. For each $e \in
E(\alpha)$, if $\alpha$ is generated on the edge, then let the
corresponding adjacency $\alpha_e$ be present at 
the leaves below edge $e$, and nowhere else. Otherwise, if $\alpha$ is
cut along the edge $e$, let the corresponding adjacency $\alpha_e$ be
absent  at the leaves below edge $e$, and be present at all
other leaves. It is easy to see that
the only most parsimonious solution for $\alpha_e$ is to create or cut
$\alpha_e$ with an SCJ operation along edge $e$. Clearly, $x \in
\sharpSPSCJ$, as there are no constraints on its adjacencies, the
number of solutions for $x'$ is the same as the number of solutions
for $x$, moreover
\begin{equation}
|x'| = O(|x| + |A| \times |T|).
\end{equation}
Therefore an $\FPRAS$ algorithm for $x'$ is also an $\FPRAS$ for $x$.
\end{proof}
We can now prove Theorem~\ref{theo:hard_FPRAS}.
\begin{proof} {\bf (Theorem~\ref{theo:hard_FPRAS}).}
From Lemma~\ref{lem:2fpras}
\begin{equation}
\sharpSPSCJ \in \FPRAS \Rightarrow \sharpFSPSCJ \in \FPRAS
\end{equation}
From Lemma~\ref{lem:almost_self_reducible},
\begin{equation}
\sharpFSPSCJ \in \FPRAS \Rightarrow \sharpFSPSCJ \in \FPAUS
\end{equation}
Putting these together, we get that
\begin{equation}
\sharpSPSCJ \in \FPRAS \Rightarrow \sharpFSPSCJ \in \FPAUS
\end{equation}
But from the proof of Theorem~\ref{theo:non-approx}. It is clear that
an $\FPAUS$ already for $\sharpFSPSCJ$ would imply that $\RP = \NP$.
\end{proof}

\section{Discussion/Conclusions}
\label{sec:conclusion}

We proved non-approximability for a counting problem motivated
by computational biology, whose optimization/decision
counterpart problem is in $\P$.

The problem is related to the evolution of discrete characters:
imagine a set of $n$ independent characters from a finite set
(a nucleotide sequence where all nucleotides evolve independently
for instance), and a set of species related by a binary phylogenetic
tree. The values of the $n$ characters are known at the leaves,
and the small parsimony problem asks for assignments at the
internal nodes of the tree. Here finding one most parsimonious
assignment is easy, but it is also easy to count their number or
sample them uniformly, when they are all independent, which is not the
case for adjacencies in genomes. However, if the assignments are
weighted by the number of most parsimonious evolutionary scenarios on
the whole set of characters, then there is no possible efficient
counting or sampling method. Indeed, in our proof all adjacencies are
independent, so it applies to this more general problem.

This study also highlights a counting bias in the parsimony $\SCJ$
model with independent adjacencies (or evolutionary scenarios on discrete characters). For example,
take a cherry with ambiguous values at its root. The number of scenarios
is higher if the assignment at the root of the cherry is equal to one of the
leaves than if it is a mix between the two. In an unbiased model
all assignments should be equiprobable. This observation leads to two possible
directions for future work:
\begin{itemize}
\item {\bf Counting assignments.} If all assignments should be equiprobable,
then the problem is to count and sample in the assignment solution space.
It is our unpublished result that counting
the number of Fitch solutions to $\SPSCJ$ is in $\FP$, but counting the Sankoff type assignments
has an unknown computational complexity.
\item {\bf Probabilistic models.} The bias of the parsimony model will
drop in a probabilistic approach. Here mutations follow a continuous
  time Markov model. In that case, each potential $\SCJ$ operation has an
  exponential waiting time for the occurrence. The so-called trajectory
  likelihood can be calculated analytically, see \cite{mlh2004}. The
  sum of the trajectory likelihoods is the total likelihood of two
  genomes, {\em i.e.}, what is the probability that genome $G_1$ becomes
  genome $G_2$ after time $t$, given a set of parameters for
  the exponential distributions put onto the potential $\SCJ$
  operations. The total likelihood calculation has an unknown
  computational complexity.

  We can also consider the probabilistic approach on a tree. In case
  of independent events, it can be shown that the multinomial
  coefficients describing how many combinations exist to merge the
  independent $\SCJ$ operations are cancelled out in the likelihood
  calculations. If all edge lengths of the evolutionary tree are the
  same, and all adjacencies are independent, then the probabilistic
  $\sharpSPSCJ$ problem reduces to counting the assignments to the
  internal nodes of the evolutionary tree, which might have a simpler
  computational complexity.
\end{itemize}

These are promising future directions of research, which can
be important for comparative genomics. To close the mathematical
aspects of the $\SPSCJ$ problem, two unsolved
questions remain:
\begin{itemize}
\item {\bf $\sharpP$-completness of $\sharpSPSCJ$.} Our conjecture is
  that $\sharpSPSCJ \in
  \sharpPcomp$. Theorem~\ref{theo:pnp} strengthens this
  conjecture. Although $\sharp3SAT \in \sharpPcomp$, the construction in the proof of
  Theorem~\ref{theo:non-approx} is not sufficient for counting the
  number of satisfying assignments of $\Phi$. For each satisfying
  assignment, there is a multiplicative coefficient that can vary
  between $\left(\frac{n-3}{2}\right)!^2$ and $(n-3)!$, and this
  shadows the exact number of solutions.
\item {\bf Star tree problem.} Given a set of genomes,
  $G_1,G_2,\ldots G_k$ related to a star tree,  count and sample their
  most parsimonious $\SCJ$ scenarios. If $k$ is odd, then the
  assignment for the centre of the star tree is unique. It is proved
  for the median genome of $3$ genomes by \cite{fm2011}, and their
  proof can be extended to any odd number of genomes. However, when
  $k$ is even, then the median might not be unique, and there might be
  exponentially many solutions for the assignment. The computational
  complexity for this case is an open question. This generalizes to the
  small parsimony problem on non-binary trees.
\end{itemize}

\section{Acknowledgments}

I.M. was supported by OTKA grant PD84297. S.Z.K. was supported by OTKA grants
K77476 and NK 105645.

\bibliographystyle{elsarticle-harv}
\bibliography{references}









\end{document}